\documentclass[12pt,titlepage]{article} 

\usepackage[utf8]{inputenc}
\usepackage{graphicx}
\usepackage{amsthm}
\usepackage{setspace}
\usepackage[a4paper]{geometry}

\usepackage{amsmath}

\usepackage[T1]{fontenc}
\usepackage[utopia]{mathdesign}

\newtheorem{theorem}{Theorem}

\newtheorem{proposition}[theorem]{Proposition}

\newtheorem{assumption}{Assumption}

\newcommand{\tr}{\text{tr }}       
\newcommand{\E}{\text{E}}       
\newcommand{\vectorization}{\text{vec}}       
\newcommand{\scr}{\mathcal}

\title{{\large FEDERAL UNIVERSITY OF MINAS GERAIS \\
Department of Electronic Engineering}
\\[2.5cm]
\Large Technical Report
\\[0.5cm]
\large Bayesian Subspace Identification in the MIMO Case
}

\author{\\[3cm] Alexandre Rodrigues Mesquita}

\begin{document}

\maketitle

\abstract{This report investigates the extension of the Bayesian Subspace System Identification method proposed in \cite{mesquita2024robust} to the Multiple-Input Multiple-Output (MIMO) case. We derive new equivariant priors and posterior distributions specifically suited for the MIMO framework. Numerical results utilizing the DAISY dataset are reported to validate the approach.}



\section{Introduction}
\label{sec:introduction}

Subspace system identification relies on estimating subspaces defined by a large number of parameters. This estimation process may benefit significantly from variance reduction techniques like regularization and Bayesian inference. In the last decade, regularization techniques such as nuclear norm minimization \cite{verhaegen2016n2sid,smith2014frequency,pillonetto2016regularized,chiuso2019system,sun2022finite} have been explored in the context of subspace identification methods. Shrinkage estimators were first applied to subspace identification in \cite{yin2021low}. While Bayesian methods were recently introduced in \cite{mesquita2024robust}, they were restricted to single-input single-output (SISO) models. 
The goal of this contribution is to extend Bayesian subspace identification to the multiple-input multiple-output (MIMO) case. 

Two key steps in subspace identification are i) some variant of linear regression and ii) the approximation of some data matrix by a low-rank matrix obtained by truncating the singular values of the original matrix. Bayesian subspace identification replaces these two steps by alternating Bayesian regressions in the row space and column space of the observed data matrix. In other words, the empirical posterior distribution is computed by a series of Gibbs sampling steps consisting of alternating regularized least squares problems.
The prior distribution is designed to obey some structural properties of the estimated matrices, such as the lower block triangular structure of the noise covariance matrix. In particular, equivariant prior distributions were derived for such matrices.

In this contribution, we derive i) an extended equivariant prior distribution that applies to the MIMO case and ii) derive the corresponding posterior distribution suited for the MIMO case. Since the desired extended result was not trivially available from the SISO case, we made use of properties of the block vectorization operation in order to find a computable posterior distribution.

\section{The subspace identification framework}
\label{sec:problem}

We assume the system whose model we want to identify obeys the linear state space model in innovation form
\begin{align}
x_{k+1} &= A x_k +B u_k +K e_k \\
y_k &= Cx_k + D u_k + e_k
\enspace,
\end{align}
where $x_k\in \mathbb{R}^{n_x}$ is the state, $y_k\in\mathbb{R}^{n_o}$ is the measured variable, $u_k\in\mathbb{R}^{n_i}$ is the input variable and $e_k$ is the innovations process, assumed to be white Gaussian noise. The model parameters $A, B, C, D$ and $K$ are matrices with the appropriate dimensions.
This same model may be expressed equivalently in in the predictor form
\begin{align}
x_{k+1} &= A_K x_k +B_K z_k \\
y_k &= Cx_k + D u_k + e_k
\enspace,
\end{align}
where $z_k=[u_k^T ~~y_k^T]^T$, $A_K=A-KC$, $B_K=[B-KD~~K]$.

Following the notation in \cite{qin2006overview}, we make use of the extended state space model
\begin{equation}
Y_f = \Gamma_f X_k + H_f U_f + G_f E_f
\label{eq:mod0}
\end{equation}
and its predictor form
\begin{equation}
Y_f = H_{fp} Z_p + H_f U_f + G_f E_f
\enspace,
\label{eq:main}
\end{equation}
where the available data is arranged in block Hankel matrices defined as
\begin{equation}
Y_f = \left[
\begin{array}{cccc}
y_k & y_{k+1} & \cdots & y_{k+N-1}\\
y_{k+1} & y_{k+2} & \cdots & y_{k+N}\\
\vdots & \vdots & \ddots & \vdots\\
y_{k+f-1} & y_{k+f} & \cdots & y_{k+f+N-2}\\
\end{array}
\right] \label{eq:hankely}
\end{equation}
and similarly are defined $U_f$ and $E_f$. As usual in the subspace identifiation literature, $f$ and $p$ are used both as labels for future and past variables as well as numbers ($f>n_x$, $p>n_x$) that denote the future and past horizons, respectively. Here $N$ is an integer that depends on the size of the available data set. The state sequence is defined as
\begin{equation}
X_k = \left[
\begin{array}{cccc}
x_k & x_{k+1} & \cdots & x_{k+N-1}
\end{array}
\right]
\enspace.
\end{equation}
The past information is collected in $Z_p = [U_p^T ~ Y_p^T]^T$ up to the horizon $p>n_x$ and arranged as
\begin{equation}
U_p = \left[
\begin{array}{cccc}
u_{k-p} & u_{k-p+1} & \cdots & u_{k-p+N-1}\\
u_{k-p+1} & u_{k-p+2} & \cdots & u_{k-p+N}\\
\vdots & \vdots & \ddots & \vdots\\
u_{k-1} & u_{k} & \cdots & u_{k+N-2}\\
\end{array}
\right]
\label{eq:hankel}
\end{equation}
and similarly for $Y_p$. 
As a consequence, $\Gamma_f$ is the extended observability matrix as defined by
\begin{equation}
\Gamma_f =
\left[
\begin{array}{c}
C \\
CA \\
\vdots\\
CA^{f-1}
\end{array}
\right]
\enspace,
\end{equation}
$H_f$ and $G_f$ are block Toeplitz matrices given by
\begin{equation}
\label{eq:hf}
H_f =
\left[
\begin{array}{cccc}
D & 0 & \cdots & 0 \\
CB & D & \cdots & 0\\
\vdots & \vdots & \ddots & \vdots\\
CA^{f-2}B & CA^{f-3}B & \cdots & D
\end{array}
\right]
\enspace,
\end{equation}
and
\begin{equation}
G_f =
\left[
\begin{array}{cccc}
I_{n_o} & 0 & \cdots & 0 \\
CK & I_{n_o} & \cdots & 0\\
\vdots & \vdots & \ddots & \vdots\\
CA^{f-2}K & CA^{f-3}K & \cdots & I_{n_o}
\end{array}
\right]
\cdot (I_f\otimes \Sigma^{1/2})
\enspace,
\label{eq:lowtoep}
\end{equation}
where $I_f$ is the $f\times f$ identity matrix and $\Sigma$ is the innovations covariance.
The matrix $H_{fp}$ is reminiscent of the matrix of Markov parameters and is given by 
\begin{equation}
H_{fp} = \left[
\begin{array}{cc}
H_{fp}^{(1)} & H_{fp}^{(2)}
\end{array}
\right]
\end{equation}
with 
\begin{equation}
H_{fp}^{(i)} = 
\left[
\begin{array}{cccc}
CA_K^{p-1}B_K^{(i)} & CA_K^{p-2}B_K^{(i)} & \cdots & CB_K^{(i)} \\
CAA_K^{p-1}B_K^{(i)} & CAA_K^{p-2}B_K^{(i)} & \cdots & CAB_K^{(i)} \\
\vdots & \vdots & \ddots & \vdots \\
CA^{f-1}A_K^{p-1}B_K^{(i)} & CA^{f-1}A_K^{p-2}B_K^{(i)} & \cdots & CA^{f-1}B_K^{(i)} \\
\end{array}
\right]
\enspace,
\end{equation}
for $i=1,2$ and $B_K^{(1)}=B-KD$ and $B_K^{(2)}=K$. These matrices can also be decomposed as products $H_{fp}^{(i)} = \Gamma_fL_p^{(i)}$ of the extended observability $\Gamma_f$ and controllability matrices:
\begin{equation}
L_p^{(i)}=[
\begin{array}{ccc}
A_K^{p-1}B_K^{(i)} & A_K^{p-2}B_K^{(i)} \cdots B_K^{(i)}
\end{array}
]
\enspace
\end{equation}
for $i=1,2$.

Building upon this formulation, numerous estimation methods have been proposed in the literature that exploit the low-dimensional subspace structure of the data as organized in
(\ref{eq:main}) (see \cite{qin2006overview,van2012subspace} for a comprehensive review). A common approach is to solve (\ref{eq:main}) via least squares, followed by estimating $\Gamma_f$ from a truncated singular value decomposition of $H_{fp}$. Once $\Gamma_f$ is determined, the system matrices $A$ and $C$ can be recovered from the relatioship between $\Gamma_f$ and $\Gamma_{f+1}$. Subsequently, the matrices $B$ and $D$ are obtained either by solving a secondary least squares problem or by utilizing estimates of the controllability matrix  $L_p$.

Generally, most subspace methods proceed by obtaining an initial estimate $\hat{H}_{fp}$ through least squares. Given two weight matrices $W_1$ and $W_2$, a low-dimensional estimate $\hat{\hat{H}}_{fp}$ is then derived by truncating the singular value decomposition of ${W_1\hat{H}}_{fp}W_2$ to its $r$ largest components:
\begin{equation}
W_1\hat{H}_{fp}W_2 = USV^T \approx U_r S_r V_r^T =: W_1\hat{\hat{H}}_{fp}W_2
\enspace.
\end{equation}
In this context, the truncation rank $r$ serves as  an estimate of the system order $n_x$.

Traditional subspace identification methods obtain an initial estimate of $H_{fp}$ from (\ref{eq:main}) using the following least squares solution:
\begin{equation}
[\hat{H}_{fp} ~\hat{H}_f] = Y_f\left[
\begin{array}{c}
U_p\\
Y_p\\
U_f
\end{array}
\right]^\dagger
\enspace.
\label{eq:LS}
\end{equation}
However, it is important to note that this estimate merely approximates the maximum likelihood estimate. This is because the noise term $G_fE_f$ in (\ref{eq:main}) is, at best, only approximately white and the structure of $H_f$ as a block lower triangular Toeplitz matrix is not taken into account.

\section{Problem Description}

Since there are many competing methods in the literature to estimate the matrices $(A,B,C,D,K)$ from $H_{fp}$ and since the estimation errors from such methods are upper bounded by a constant times $\|H_{fp}-\hat{H}_{fp}\|$ as demonstrated in \cite{tsiamis2019finite}, we concentrate our efforts on the estimators of $H_{fp}$. With this in mind, and following the statistical theory of point estimation \cite{robert2007bayesian}, we seek estimators that perform well with regards to risk functions of the form
\begin{equation}
\mathsf{R}_\theta(\hat{\hat{H}}_{fp}) = \text{E}\left[\|\left. W_1(H_{fp}-\hat{\hat{H}}_{fp})W_2\|_F^2) \right| (A,B,C,D,K)=\theta \right]
\enspace.
\label{eq:risk}
\end{equation}
The risk is a function of the estimator being used and also a function of the system parameters $\theta$. The expectation is over all input-output realizations for the system model with parameter $\theta$. In general, it is not possible to minimize the risk for all $\theta$. Therefore one often seeks to minimize the expectation of the risk over some distribution of $\theta$ or the risk for the worst case of $\theta$.

Within this context, our goal is to investigate robustified methods to estimate $H_{fp}$. In particular, we want to design Bayesian estimators that can be efficiently computed with a Gibbs sampling scheme.

\section{An Alternating Least Squares Bayesian Approach}

Bayesian estimators are inherently robust. Under mild conditions \cite{robert2007bayesian}, they are admissible, i.e., there exists no other estimator that improves upon their risk for all the parameters in the parameter space. Under structural conditions involving group invariance of priors and loss functions (see \cite{robert2007bayesian} for definitions), they are also minimax, i.e., they minimize the risk under the worst case parameter value.

In our approach to subspace identification, we aim to construct a Bayesian method that is computationally simple. With this in mind, we choose priors that lead to simple regularized least squares steps. The priors are obtained empirically from the data. This Bayesian model was inspired by that in \cite{ding2011bayesian} but, differently from this paper, we do not explicitly compute singular value decompositions with the orthogonal matrices $U$ and $V$ and the singular value matrix $S$. We found that better results are obtained by directly estimating $US^{1/2}$ and $S^{1/2}V^T$ instead.

\begin{assumption}
Suppose that $f=i$ and $N=j$ so the Hankel matrix $Y_f$ is of size $in_o\times j$. We adopt the following set of conditionally independent priors:
\begin{align}
\Gamma_f &= \bar{G}_f \Xi_\Gamma \Lambda_\Gamma^{-1/2} \label{eq:gammaprior}\\
H_f &= \bar{G}_f \Xi_H \Lambda_H^{-1/2} \\
L_p &= \Lambda_L^{-1/2} \Xi_L Z_p^\dagger \label{eq:lpprior}\\
G_f &\sim \frac{1}{|\mathsf{G}_{11}|^{in_o}}\label{eq:Gprior}
\end{align}
where the $\Xi_{(\cdot)}$ matrices are random matrices with the appropriate dimensions and whose elements are independent, normally distributed random variables with variance $1$, the $\Lambda_{(\cdot)}$ matrices are fixed parameters to be specified, $\bar{G}_f:=G_f/|\mathsf{G}_{11}|^i$, and $G_f$ is restricted to the set of block lower triangular Toeplitz matrices.
\end{assumption}

To define the parameters in the priors, we initially obtain estimates ${H}_{fp}^{(1)}$ and ${H}_f^{(1)}$ from (\ref{eq:LS}) and next compute a truncated singular value decomposition $H_{fp}^{(1)}Z_p \approx U_rS_rV_r^T$. Then, we make $\bar{G}_f^{(1)}=I_i$ and
\begin{align}
\Gamma_f^{(1)} &= U_rS_r^{1/2} \label{eq:gamma1} \\  
L_p^{(1)} &=S_r^{1/2}V_r^TZ_p^\dagger \label{eq:Lp1} \\ 
\Lambda_\Gamma^{-1} &=S_r/(in_o) \label{eq:Lambda}\\
\Lambda_L^{-1}  &=S_r/j \\
\Lambda_H^{-1} &= I_{in_i}\tr( (H_f^{(1)})^TH_f^{(1)})/(i^2n_on_i) \label{eq:lambdaH}
\enspace.
\end{align}
To justify such a choice, we replace the above values of $\bar{G}_f$, $\Lambda_\Gamma^{-1}$ and $\Lambda_L^{-1}$ in (\ref{eq:gammaprior}) and (\ref{eq:lpprior}) and note that they describe the approximate SVD decomposition:
\begin{equation}
\Gamma_f X_p = \frac{\Xi_\Gamma}{\sqrt{in_o}} S_r\frac{\Xi_L}{\sqrt{j}}
\enspace,
\end{equation}
where ${\Xi_\Gamma}/\sqrt{in_o}$ and ${\Xi_L}/\sqrt{j}$ behave as orthogonal matrices in expectation, namely, $\E[\Xi_\Gamma^T\Xi_\Gamma/{in_o}]= \E[\Xi_L\Xi_L^T/{j}] = I_r$. With this, we have a prior that both reflects the desired SVD decomposition and that allows for simple posterior computations. 

The prior distribution for the noise term $G_f$ exploits its block lower triangular Toeplitz structure as exemplified in (\ref{eq:lowtoep}). The set of block lower triangular non-singular Toeplitz matrices is a group under matrix multiplication. This group operation may be interpreted as the cascading of dynamical systems. In this sense, when we impose invariance with respect to this group, we are interested in estimators that are consistent with respect to dynamical system cascading. For example, if we pass both input and output through a linear filter, we want an estimator that gives the same model regardless of the filtering. In the Bayesian framework, such an equivariant estimator demands a prior distribution that is invariant under the group action \cite{robert2007bayesian}. This property is the subject of the following theorem. 

\begin{theorem}
\label{thm:equivariance}
The improper prior probability distribution given in (\ref{eq:Gprior}) for the block lower triangular Toeplitz matrix $G_f$ is invariant under the group actions of left and right multiplication by non-singular matrices of the same structure, i.e., $\mathsf{A}G_f$ and $G_f\mathsf{B}$ have this same distribution for any fixed block lower triangular Toeplitz non-singular matrices $\mathsf{A}$ and $\mathsf{B}$. 
\end{theorem}

Because the full posterior distribution for the estimation problem at hand is too hard to characterize analytically, we estimate its empirical distribution using a Gibbs sampler. In a Gibbs sampler, samples from dependent variables $x$ and $y$ are drawn iteratively from their conditional distributions as in: \begin{align}
&y^{(n)}\sim p(y|x^{(n-1)}) \\
&x^{(n)}\sim p(x|y^{(n)})
\enspace.
\end{align}
The probability distribution of the resulting Markov chain $(x^{(n)},y^{(n)})$ is shown to converge under mild conditions to the joint distribution $p(x,y)$. 

Applying our given priors to (\ref{eq:mod0}) with $X_p$ and $G_f$ known, we have a Gaussian linear model whose posterior can be easily computed. On the other hand, if we regard $\Gamma_f, H_f$ and $G_f$ as known, we have a second Gaussian linear model whose posterior is well-known. Such posteriors are summarized in the next proposition.

\begin{proposition}
\label{prop:posterior}
Assume \emph{a priori} distributions given by (\ref{eq:gammaprior})-(\ref{eq:lpprior}). Then, conditioned on $X_p^{(n-1)}=L_p^{(n-1)}Z_p$ and $G_f^{(n-1)}$, samples from the posterior distribution of $\Gamma_f$ and $H_f$ are obtained from:
\begin{equation}
\label{eq:regressionGammaH}
\left[
\begin{array}{cc}
\Gamma_f^{(n)} & H_f^{(n)}
\end{array}\right]
 = Y_f
 \left[
 \begin{array}{c}
 X_p^{(n-1)} \\
 U_f
\end{array}  
\right]^T
\left(
\Sigma_1^{(n)}
\right)^{-1}
\gamma^{(n)}
\\+
\bar{G}_f^{(n-1)}\Xi_{\Gamma,H}^{(n)}\left(
\Sigma_1^{(n)}\right)^{-1/2}
\end{equation}
\begin{equation}
\label{eq:Sigma1}
\Sigma_1^{(n)} = 
 \left[
 \begin{array}{c c}
 \Lambda_\Gamma & 0\\
 0 & \Lambda_H \\
\end{array}  
\right]+
 \left[
 \begin{array}{c}
 X_p^{(n-1)} \\
 U_f
\end{array}  
\right]
 \left[
 \begin{array}{c}
 X_p^{(n-1)} \\
 U_f
\end{array}  
\right]^T
\gamma^{(n)}
\enspace, 
\end{equation}
where $\gamma^{(n)}=1/|\mathsf{G}_{11}^{(n-1)}|^{2i}$ and $\Xi_{\Gamma,H}^{(n)}$ is a random matrix whose components are independently drawn from a unit normal distribution. 

Moreover, conditioned on $\Gamma_f^{(n)}$ and $H_f^{(n)}$, we obtain samples from the posterior distribution of $L_p$ with the following equations:

\begin{equation}
L_p^{(n)}
 = 
\left(\Sigma_2^{(n)}\right)^{-1}
\Gamma_f^{(n)T}\Psi_e^{(n)}
(Y_f-H_f^{(n)}U_f)Z_p^\dagger 
+
\left(\Sigma_2^{(n)}\right)^{-1/2}\Xi_L^{(n)}Z_p^\dagger
\label{eq:reg2}
\end{equation}
\begin{equation}
\label{eq:sigma2}
\Sigma_2^{(n)}
 = 
\Gamma_f^{(n)T}\Psi_e^{(n)}\Gamma_f^{(n)}
+\Lambda_L
\enspace,
\end{equation}
where $\Psi_e^{(n)} = \left(G_f^{(n-1)}(G_f^{(n-1)})^T\right)^{-1}$ and $\Xi_{L}^{(n)}$ is a random matrix whose components are independently drawn from a unit normal distribution.
\end{proposition}

Intuitively, we are performing independent regressions row by row in (\ref{eq:regressionGammaH}) and independent regressions column by column in (\ref{eq:reg2}), in addition to summing the corresponding simulated noise terms. 

\begin{theorem}
\label{thm:postG}
Let $\mathcal{T}_{i\times i}$ be the matrix that maps the first column of a lower triangular Toeplitz matrix of dimension $i\times i$ to its vectorization, i.e., $\vectorization(\mathcal{G}) = \mathcal{T}_{i\times i} (\mathcal{G}_{[:,1]})$ for any lower triangular Toeplitz matrix $\mathcal{G}$ of row-length $i$. Let $\mathcal{H}_{i\times j}$ be the matrix that maps an $i+j-1$ vector into an $i\times j$ Hankel matrix such as in (\ref{eq:hankely}). Let $\mathcal{W}_{d}(\mathsf{V},\mathsf{n})$ denote the $d$-dimensional Wishart distribution of $\mathsf{n}$ degrees of freedom and scale matrix $\mathsf{V}$. Then, conditioned on $\Gamma_f, L_p$ and $H_f$ known, a sample from the posterior distribution of $G_f$ is obtained from the equations:
\begin{align}
&\mathcal{E}^{(n)} = Y_f - \Gamma_f^{(n)}L_p^{(n)}Z_p -H_f^{(n)}U_f \label{eq:posti}\\
&\Omega^{(n)} = (\mathcal{T}_{i\times i}^T\otimes I_{n_o})(I_i\otimes\mathcal{E}^{(n)})\bar{\mathcal{H}}_{j\times i}(I_{n_o}\otimes(\mathcal{E}^{(n)})^T )(\mathcal{T}_{i\times i}\otimes I_{n_o}) \label{eq:Omega}\\
&\Theta^{(n)} \sim \mathcal{W}_{n_o}(I_{n_o},i+j-in_o) \\
&\mathsf{N}_1^{(n)} =(\Theta^{(n)})^{1/2} \\
&\mathsf{N}_\ell^{(n)} \sim N(0,I_{n_o}) , ~\text{for} ~\ell=2, \ldots, i\\
&(G_f^{(n)})^{-1}_{[1:in_o,1:n_o]} = (\Omega_U^{(n)})^{-1} [\mathsf{N}_1 ~ \mathsf{N}_2 \cdots \mathsf{N}_i]^T \label{eq:post1}
\enspace,
\end{align}
where $\Omega_U^{(n)}$ denotes the upper triangular part of the Cholesky decomposition of $\Omega^{(n)}$ and $\bar{\mathcal{H}}_{j\times i}=\mathcal{H}_{j\times i}(\mathcal{H}_{j\times i}^T\mathcal{H}_{j\times i})^{-2}\mathcal{H}_{j\times i}^T$.

Furthermore, if one adopts the approximation that the noise vector $e_f=\vectorization(E_f)$ has independently distributed components, then, the conditional  posterior distribution of $G_f$ is obtained from:
\begin{align}
&\mathcal{E}^{(n)} = Y_f - \Gamma_f^{(n)}L_p^{(n)}Z_p -H_f^{(n)}U_f \label{eq:posti2}\\
&\Omega^{(n)} = (\mathcal{T}_{i\times i}^T\otimes I_{n_o})(I_i\otimes\mathcal{E}^{(n)}(\mathcal{E}^{(n)})^T )(\mathcal{T}_{i\times i}\otimes I_{n_o}) \\
&\Theta^{(n)} \sim \mathcal{W}_{n_o}(I_{n_o},ij-in_o+1) \\
&\mathsf{N}_1^{(n)} =(\Theta^{(n)})^{1/2} \\
&\mathsf{N}_\ell^{(n)} \sim N(0,I_{n_o}) , ~\text{for} ~\ell=2, \ldots, i\\
&(G_f^{(n)})^{-1}_{[1:in_o,1:n_o]} = (\Omega_U^{(n)})^{-1} [\mathsf{N}_1 ~ \mathsf{N}_2 \cdots \mathsf{N}_i]^T \label{eq:posti22}
\enspace.
\end{align}

\end{theorem}

\begin{proof}

The following preliminary result is needed regarding the block vectorization of block matrices. 

Let $P, Q$ and $R$ be block matrices such that $P=QR$ and let $P$ and $R$ be organized in columns 
of block matrices such that $P=\left[\mathsf{P}_{:,1}~ \mathsf{P}_{:,2} \cdots \mathsf{P}_{:,K}\right]$ and $R=\left[\mathsf{R}_{:,1}~ \mathsf{R}_{:,2} \cdots \mathsf{R}_{:,K}\right]$. By stacking these columns we define the block vectorization $\text{vec}_b(P)$ and $\text{vec}_b(R)$. Now, since $\mathsf{P}_{:,k}=Q\mathsf{Q}_{:,k}$, we have that
\begin{equation}
\text{vec}_b(P)=(I_K\otimes Q)\text{vec}_b(R)
\enspace.
\end{equation}
Given this preliminary result, we now proceed to consider the residues
\begin{equation}
\mathcal{E} = Y_f-H_{fp}Z_p-H_fU_f = G_f E_f
\enspace.
\end{equation}
Taking the vectorization operation and maki	ng use of the noise vector $\bar{e}_f$ on the subspace of dimension $n_o(i+j-1)$, we have that
\begin{equation}
\vectorization(\mathcal{E}) = (I_j\otimes G_f) (\mathcal{H}_{i\times j}\otimes I_{n_o})\bar{e}_f
\enspace.
\end{equation}
Therefore, the residues covariance is given by
\begin{equation}
\Sigma_{\mathcal{E}} = (I_j\otimes G_f) (\mathcal{H}_{i\times j} \mathcal{H}_{i\times j}^T \otimes I_{n_o})(I_j\otimes G_f^T)
\enspace.
\end{equation}
Since $\text{rank}(\mathcal{H}_{i\times j})=(i+j-1)$, $\Sigma_{\mathcal{E}}$ is rank deficient and we shall make use of its pseudo-determinant $|\cdot|_+$ (product of non-zero singular values) in obtaining its pdf:

\begin{align}
|\Sigma_{\mathcal{E}}|_{+} &= |(I_j\otimes G_f) (\mathcal{H}_{i\times j} \mathcal{H}_{i\times j}^T \otimes I_{n_o}) (I_j\otimes G_f^T)|_{+}  \\
&= |(I_j\otimes G_f)\cdot (\text{chol}(\mathcal{H}_{i\times j} \mathcal{H}_{i\times j}^T)\otimes I_{n_o})|_{+}^2
\propto \left|\mathsf{G}_{ii}\right|^{2(i+j-1)}
\enspace,
\end{align}
where we used the fact that $G_f$ is block lower triangular Toeplitz and the fact that the lower triangular matrix $\text{chol}(\mathcal{H}_{i\times j}\mathcal{H}_{i\times j}^T)$ must have exactly $(i+j-1)$ nonzero entries on its diagonal.

In order to make the block vectorization of $G_f$ explicit in the likelihood function, we note that
\begin{align}
E_f = G_f^{-1}\mathcal{E} \Rightarrow \vectorization_b(E_f^T) = (I_i\otimes \mathcal{E}^T )\vectorization_b(G_f^{-T})
\enspace.
\end{align}
Now, let $\tilde{e}_f$ be the noise vector organized as a $i+j-1\times n_o$ matrix. Then, $\text{vec}_b(E_f^T)=\mathcal{H}_{j\times i}\tilde{e}_f$, which, combined with the previous equation, gives 
\begin{equation}
\tilde{e}_f = \mathcal{H}_{j\times i}^\dagger(I_i\otimes \mathcal{E}^T ) \vectorization_b(G_f^{-T}) = \mathcal{H}_{j\times i}^\dagger(I_i\otimes \mathcal{E}^T ) (\mathcal{T}_{i\times i}\otimes I_{n_o}) \mathsf{M}_{\cdot,1}
\enspace,
\end{equation}
where $\mathsf{M}_{\cdot,1}$ denotes the first block column of $G_f^{-T}$.

Therefore,
{\small
\begin{equation}
p(\mathcal{E}|G_f) \propto \frac{1}{\left|\mathsf{G}_{ii}\right|^{i+j-1}}
\exp\left(-\frac{1}{2}\tr \mathsf{M}_{\cdot,1}^T(\mathcal{T}_{i\times i}^T\otimes I_{n_o})(I_i\otimes \mathcal{E}) \cdot
(\mathcal{H}_{j\times i}^\dagger)^T\mathcal{H}_{j\times i}^\dagger(I_i\otimes \mathcal{E}^T)(\mathcal{T}_{i\times i}\otimes I_{n_o}) \mathsf{M}_{\cdot,1}\right)
\label{eq:likelihood}
\end{equation}
}
and the posterior would be proportional to
{
\begin{align}
p(\mathcal{E}|M)d\mu(M) &\propto \frac{|\mathsf{M}_{11}|^{i+j-1}}{|\mathsf{M}_{11}|^{in_o}} \exp\left(-\frac{1}{2}\tr \mathsf{M}_{\cdot,1}^T\Omega\mathsf{M}_{\cdot,1}\right) 
\\
&= |\mathsf{M}_{11}|^{j+i-in_o-1}\exp\left(-\frac{1}{2}\tr \mathsf{M}_{\cdot,1}^T\Omega\mathsf{M}_{\cdot,1}\right)
\enspace.
\end{align}
}
Defining the change of variables ${N}=\Omega_U\mathsf{M}_{\cdot,1}$, the triangular structure of $\Omega_U$ gives that $|\mathsf{M}_{11}|\propto|\mathsf{N}_1|$ and we have
\begin{equation}
p({N}|\scr{E}) \propto |\mathsf{N}_1|^{j+i-in_o-1}\exp\left(-\frac{1}{2}\sum_{\ell=1}^i\tr \mathsf{N}_\ell^T\mathsf{N}_\ell\right)
\enspace. 
\end{equation}
Therefore, for $\ell=2,\ldots,i$, $\mathsf{N}_\ell$ are matrices whose elements are independent and normally distributed with unit variance. To determine the distribution of $\mathsf{N}_1$, we make the change of variables $\Theta=\mathsf{N}_1^T\mathsf{N}_1$. This transformation gives a Jacobian determinant of $2|\mathsf{N}_1|^{n_o}$, which leads to the distribution
\begin{equation}
p({\Theta}|\scr{E}) \propto \frac{|\mathsf{N}_i|^{j+i-in_o-1}}{|\mathsf{N}_i|^{n_o}}\exp\left(-\frac{1}{2}\tr \Theta\right) = {|\Theta|^{\frac{j+i-in_o-1-n_o}{2}}}\exp\left(-\frac{1}{2}\tr \Theta\right) = \scr{W}_{n_o}(I_{n_o},j+i-in_o)
\enspace. 
\end{equation}

This is precisely the posterior given by equations (\ref{eq:posti})-(\ref{eq:post1}).

If, on the other hand, we assumed that $e_f\sim N(0,I_{in_oj})$, this would equivalent to replacing $\mathcal{H}_{i\times j}$ above by $I_{in_oj}$, which would lead to
\begin{equation}
|\Sigma_{\mathcal{E}}|\propto |\mathsf{G}_{ii]}|^{2ij}
\end{equation}
and the same steps as above would lead to $\Omega = (\mathcal{T}_{i\times i}^T\otimes I_{n_o})(I_i\otimes\scr{E}\scr{E}^T)(\mathcal{T}_{i\times i}\otimes I_{n_o})$ and $p(\Theta|\scr{E}) = \scr{W}_{n_o}(I_{n_o},ij-in_o+1)$. The posterior defined by equations (\ref{eq:posti2})-(\ref{eq:posti22}) would then follow. 
\end{proof}

In other words, we first compute the residues and obtain a sample of the first block column of $G_f^{-1}$, then we construct the full block lower triangular Toeplitz matrix $G_f^{-1}$ from this column and compute its inverse. Given that the matrix $\mathcal{H}_{j\times i}$ is very large and sparse, (\ref{eq:Omega}) may be somewhat tricky to compute. A safer alternative may come from ignoring the Hankel structure of $E_f$ and assuming that its elements are mutually independent, which gives the posterior updates in (\ref{eq:posti2})-(\ref{eq:posti22}). 

So far we have not taken advantage of the lower triangular Toeplitz structure of $H_f$ (as given in (\ref{eq:hf})). In order to rigorously do so, one would need to compute the posterior distribution for the  vectorization of $H_f$. However, we found that this does not tend to perform well numerically. Instead, it proved beneficial to use the following approximate method. Apply Proposition \ref{prop:posterior} to compute only the last block row of $H_f$, then reconstruct the full block lower triangular Toeplitz matrix from its last block row. The rationale here is that the regressions are row-wise independent for our assumed prior. However, with the Toeplitz structure, the elements of the last row become dependent on the information from the remaining rows and our posterior computation is only approximate. To compensate for that, we add an extra Gibbs step to recompute $\Gamma_f$ now with $H_f$ fixed:

\begin{equation}
\Gamma_f^{(n)} 
 = (Y_f-H_f^{(n)}U_f)
( X_p^{(n-1)})^T
\left(
\Sigma_\Gamma^{(n)}
\right)^{-1}
\gamma^{(n)}
\\+
\bar{G}_f^{(n-1)}\Xi_{\Gamma}^{(n)}\left(
\Sigma_\Gamma^{(n)}\right)^{-1/2}
\label{eq:reg1}
\end{equation}
\begin{equation}
\label{eq:sigmagamma}
\Sigma_\Gamma^{(n)} = 
 \Lambda_\Gamma 
+
 X_p^{(n-1)} 
(X_p^{(n-1)})^T
\gamma^{(n)}
\enspace. 
\end{equation}

Combining the conditional posterior updates given by Proposition \ref{prop:posterior}, Theorem \ref{thm:postG} and the above equations, we are able to apply the Gibbs procedure and estimate the full posterior distribution of our model. Finally, the estimate for ${H}_{fp}$ is obtained by averaging over the trajectory of the Markov chain:
\begin{equation}
\hat{H}_{fp} = \frac{1}{N_F-N_o} \sum_{n=N_o}^{N_F} \Gamma_f^{(n)}L_p^{(n)}
\enspace,
\end{equation} 
where $N_o$ is some burn-in period intended to remove the effect of transients. In order to reduce variance and improve convergence, one may prefer to average over the expected values of (\ref{eq:reg1}) and (\ref{eq:reg2}) in every step (obtained by setting the respective $\Xi_{(\cdot)}$ matrices to zero):
\begin{equation}
\label{eq:Hestimate}
\hat{H}_{fp} = \frac{1}{2(N_F-N_o)} \sum_{n=N_o}^{N_F} \E[\Gamma_f^{(n)}]L_p^{(n-1)}+\Gamma_f^{(n)}\E[L_p^{(n)}]
\enspace.
\end{equation}

\section{Experiments on the DAISY datasets}

In this section we report the results of applying the methods proposed in \cite{mesquita2024robust} to the experimental data in the DAISY datasets \cite{de1997daisy}. In doing so, the extension provided in this report is needed in order to deal with the MIMO systems.

In the present case we omit the comparison with kernel methods as the available options in the MIMO case are still incipient and we found no candidate for a fair comparison. 

As the comparison metric, we follow [18] in using the normalized one-step ahead prediction error on a validation data set of $N_{val}$ samples:
\[
\frac{1}{n_o}\sum_{m=1}^{n_o}\frac{\sum_{k=0}^{N_{val}}(\hat{y}_m[k|k-1]-y_m[k])^2}{\sum_{k=0}^{N_{val}} y_m[k]^2}
\enspace,
\]
where $y[k]$ has been detrended.

As there are many methods to estimate the system matrices from $\hat{H}_{fp}$, we tested multiple methods for a given $\hat{H}_{fp}$ and selected the one that minimized the prediction error on the validation set. The tested procedures were the robust identification algorithm from [18] and the two methods presented in  [17].

Given that the studied methods are suited to operate with noisy data, we added different levels of extra noise to the real data. Having estimated the covariance matrix $\Sigma_y$ of the output, we generated white Gaussian noise $w[k]\sim N(0,\alpha^2\Sigma_y)$ and colored noise with the equation $v[k]=0.5v[k-1]+\sqrt{3}/2 w[k]$, where the noise level parameter $\alpha$ is fixed. Both estimation and validation data were then contaminated with $v[k]$.

The row-length of Hankel matrices was set to $i= \min\{15,\lfloor N/(10(n_o+n_i))\rfloor\}$. The weight $W_2$ was set to $Z_p$ as in the paper and the weight matrix $W_1$ was set to
\[
W_1 = I_i \otimes \Lambda_{y}
\enspace,
\]
where $\Lambda_{y}$ is a diagonal matrix constructed by taking the inverse of the square root of the diagonal elements of $\Sigma_y$.

The adopted sample and validations sizes are informed in Table \ref{tab:samplesizes}.

\begin{table}
\centering
\begin{tabular}{|l|c|c|}
\hline 
Dataset & $N$ & $N_{val}$ \\
\hline
Heating system & 280 & 521 \\
Powerplant & 120 & 80 \\
Winding & 875 & 1625 \\
Wall thermal resistance & 1260 & 420\\
Reactor & 2625 & 4875 \\
Glass furnace & 436 & 811\\
Evaporator & 2206 & 4099\\
Industrial dryer & 303 & 564 \\
Ball beam$^*$ & 1000 & 1000 \\
Hair dryer & 350 & 650 \\
CD player arm & 716 & 1332\\
Flutter$^*$ & 1024 & 1024\\
Robot arm & 768 & 256\\
Flexible structure & 2983 & 5540\\
\hline
\end{tabular}
\caption{Estimation sample size $N$ and validation sample size $N_{val}$ adopted in each dataset of the DAISY database. ($^*$): estimation and validation data are the same.}
\label{tab:samplesizes}
\end{table}

The variation on the average normalized prediction error as a function of the noise level is depicted in Figure \ref{fig:pevsnr}. The results are qualitatively similar to those in \cite{mesquita2024robust} with up to $25$\% reduction in the average prediction error. As the best performance is achieved by OptShrink, we infer that most of the noise is output noise. 

  \begin{figure}[thpb]
  \begin{center}
  \resizebox{150mm}{!}{\includegraphics{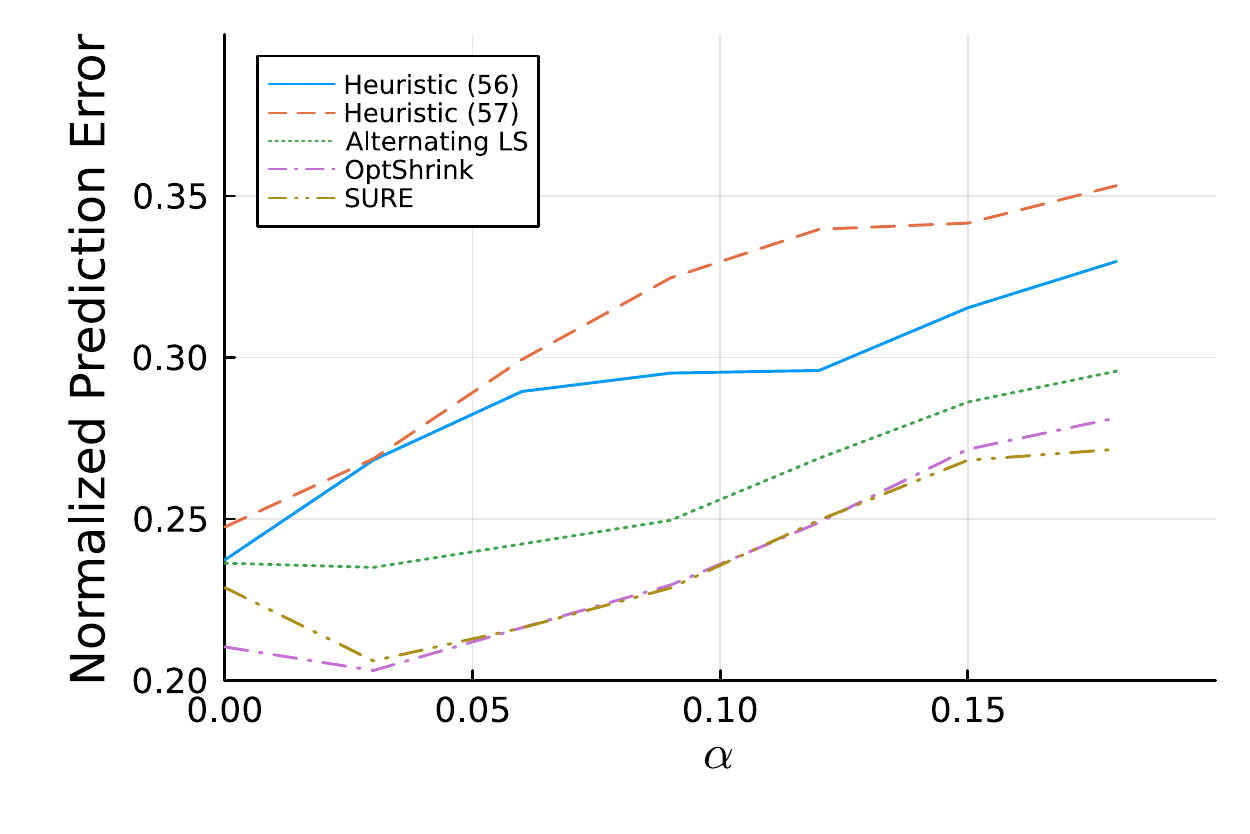}}  
  \end{center}  
  \caption{Average normalized prediction error as a function of the added noise level $\alpha$ for all $14$ datasets in the DAISY database.}
  \label{fig:pevsnr}
  \end{figure}

In Figure \ref{fig:boxplot} we see the overall statistics of the experiments on the DAISY datasets. We observe that OptShrink and SURE give the least risk in the median, third quartile and worst case. All the proposed methods improved performance in the third quartile and in worst case.

  \begin{figure}[thpb]
  \begin{center}
  \resizebox{150mm}{!}{\includegraphics{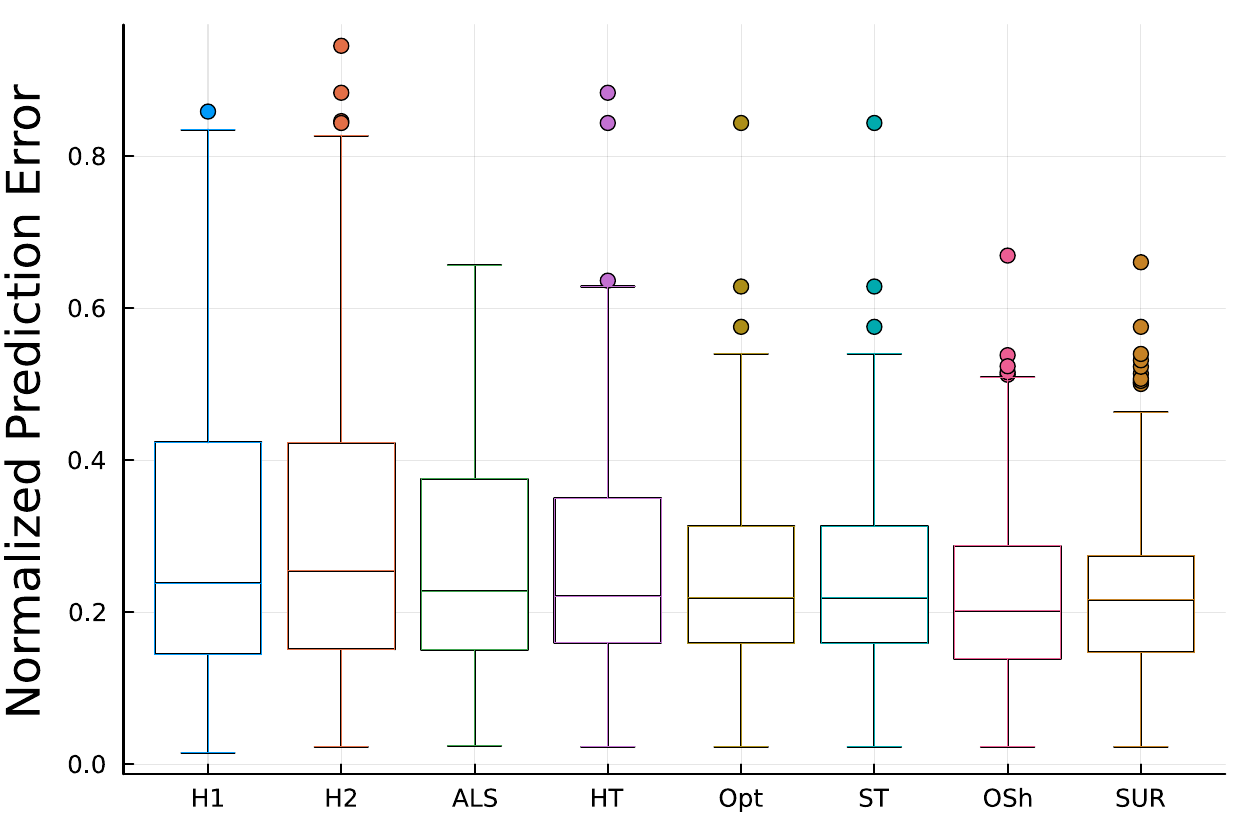}}  
  \end{center}  
  \caption{Boxplot of the normalized prediction error obtained by different estimation methods applied to the $14$ datasets in the DAISY database with different levels of additive noise.}
  \label{fig:boxplot}
  \end{figure}

\bibliographystyle{unsrt}
\bibliography{RobustBayes}

\end{document}